\documentclass[12pt]{article}
\usepackage{graphicx}
\usepackage{amsmath}
\usepackage{amsfonts}
\usepackage{amssymb}
\newtheorem{theorem}{Theorem}[section]

\newtheorem{lemma}[theorem]{Lemma}

\newenvironment{proof}[1][Proof]{\textsc{#1.} }{\ \rule{0.5em}{0.5em}}
\numberwithin{equation}{section}

\begin{document}

\title{Conformal classes of
asymptotically flat, static vacuum data.
}

\author{Helmut Friedrich\\ 
Max-Planck-Institut f\"ur Gravitationsphysik\\
Am M\"uhlenberg 1\\
14476 Golm, Germany}

\maketitle

\begin{abstract}
\noindent
{\footnotesize
We show that time-reflection symmetric,
asymptotically flat, static vacuum data
which admit a non-trivial conformal rescaling which leads again to such
data must be axi-symmetric and admit a conformal Killing field. Moreover, it
is shown that there exists a $3$-parameter family of such data. 
}
\end{abstract}


{\footnotesize

\section{Introduction}

The metric of a static space-time assumes in suitable local coordinates
$t$, $\tilde{x}^a$, $a = 1, 2, 3$, the form
\begin{equation}
\label{staticlineel}
\tilde{g} = v^2\,d\,t^2 + \tilde{h},\quad \quad
v = v(\tilde{x}^c) > 0,\quad \quad
\tilde{h} =  \tilde{h}_{ab}(\tilde{x}^c)\,d\tilde{x}^a\,d\tilde{x}^b, 
\end{equation}
where $\tilde{h}$ denotes a negative definite metric on the time slices 
$\tilde{S}_c = \{t = c = const.\}$. The hypersurface orthogonal, time-like
Killing vector field is then given by $\partial_t$. We refer to 
$\tilde{h}$ as the {\it static metric} and to $v$ as the {\it potential}. 
Einstein's vacuum
field equations reduce here to the {\it static vacuum field
equations}
\begin{equation}
\label{statEinstvac}
R_{ab}[\tilde{h}] = \frac{1}{v}\,\tilde{D}_a\,\tilde{D}_b\,v,
\quad\quad
\Delta_{\tilde{h}}\,v = 0, 
\end{equation}
where $\tilde{D}$ denotes the covariant derivative defined by $\tilde{h}$.
It suffices to consider these equations on 
$\tilde{S} \equiv \tilde{S}_0$.

In the following we study solutions which are asymptotically flat so
that the coordinates $\tilde{x}^a$ can be required, for suitable  
$\tilde{S}$, to map the set
$\tilde{S}$  diffeomorphically onto the complement of a closed ball $B_R(0)$
in $\mathbb{R}^3$ so that the fields $\tilde{h}$, $v$ 
satisfy with some $\epsilon > 0$ and $k \ge 2$ 
the asymptotic flatness 
condition\footnote{The terms $O_k(|\tilde{x}|^{-(1 + \epsilon)})$ behave
like 
$O(|\tilde{x}|^{-(1 + \epsilon + j)})$ under differentiations of order $j \le
k$. }
\begin{equation}
\label{falloff}
\tilde{h}_{ac} = - \left(1 + \frac{2\,m}{|\tilde{x}|}\right)\,\delta_{ac} 
+ O_k(|\tilde{x}|^{-(1 + \epsilon)}),
\quad 
v = 1 - \frac{m}{|\tilde{x}|} + O_k(|\tilde{x}|^{-(1 + \epsilon)})
\quad \mbox{as} \quad |\tilde{x}| \rightarrow \infty,
\end{equation}
where $|\,.\,|$ denotes the standard Euclidean norm. We refer to 
$(\tilde{S}, \tilde{h}, v)$ as {\it static vacuum data}. 
The coefficient $m$ represents its ADM mass. In this article we shall
only be interested in data with mass $m \neq 0$.

The pair  $(\tilde{S}, \tilde{h})$ defines a particular type of 
time-reflection symmetric initial data for Einstein's vacuum field
equations.  Our interest in such initial data is motivated by observations
which suggest that evolutions in time of time-reflection symmetric,
asymptotically flat vacuum data admit  at null infinity conformal extensions
of prescribed smoothness if and only if the data behave  in terms of a
certain type of expansion at space-like infinity up to a certain order
like static data (cf. \cite{friedrich:cargese} for a detailed discussion).

The full analysis of this relation requires detailed information on the
asymptotic behaviour of static data and their conformal structures at
space-like infinity. In previous work (\cite{friedrich:statconv}) we have
given a complete characterization of static vacuum data in terms of a
minimal set of symmetric trace free tensors at space-like infinity,
referred to as {\it null data}. That article clarifies, in particular, the
convergence problem. In the present article we take a first step towards
characterizing conformal structures of static data. Obviously, such a
characterization should be more easy if static data are
related in a one to one fashion to their conformal structures.
It will be shown that the relation is in fact more complicated.

A precise description of our result requires technical 
details.
To keep the discussion short we refer the reader to \cite{friedrich:statconv}
for further details.  Beig and Simon (\cite{beig:simon})
showed under certain assumptions, weakened later by 
Kennefick and O'Murchadha (\cite{kennefick:o'murchadha}), that
static data with $m \neq 0$ admit conformal rescalings 
\begin{equation}
\label{tildeh-h}
\tilde{h}_{ab} \rightarrow h_{ab} = \Omega^2\,\tilde{h}_{ab},
\end{equation}
with positive conformal factors $\Omega$ so that the set 
$S = \tilde{S} \cup \{i\}$, obtained by adjoining to $\tilde{S}$ a point $i$
representing space-like infinity,
acquires a real analytic differentiable structure in which the
conformal metric $h$ extends to a real analytic
metric on $S$ and
$\Omega$ extends to a function in $C^2(S) \cap C^{\omega}(\tilde{S})$ so that
\begin{equation}
\label{0-D-DD_Omega-at-i}
\Omega > 0 \quad \mbox{on} \quad \tilde{S},
\quad \quad
\Omega = 0, \quad d\Omega = 0, \quad Hess_h \Omega = - 2\,h
\quad \mbox{at} \quad i. 
\end{equation}
A particular example of such a conformal factor, determined by the
static data themselves, is given by 
\begin{equation}
\label{preferred-gauge}
\Omega = \left(\frac{1 - v}{m}\right)^2.
\end{equation}

Let $(\tilde{S}, \tilde{h}, v)$ be a static vacuum data set with
$m \neq 0$ and the fields $h$, $\Omega$ on
$S$ related to it by (\ref{tildeh-h}), (\ref{preferred-gauge}). We say
that {\it this set admits a non-trivial conformal rescaling} if there exists 
another static vacuum data set $(\tilde{S}', \tilde{h}', v')$ with
associated fields $h'$, $\Omega'$ on $S'$ so that the following holds. 
Some neighbourhood $U'$ of $i'$ in $S'$ can be identified diffeomorphically
with some neighbourhood $U$ of $i$ in $S$, identifying $i'$ with $i$, so
that after identification there exists a {\it non-constant} smooth function
$\vartheta > 0$ on $U$ with
\begin{equation}
\label{h'-h-rel}
h' = \vartheta^4\,h.
\end{equation} 
In other words, we require the existence of a conformal diffeomorphism
which maps a neighbourhood of space-like infinity with respect to
$\tilde{h}$ onto such a neighbourhood with respect to $\tilde{h}'$ so that
it extends in our gauge smoothly to a conformal map which maps $i$ onto
$i'$. The metrics $\tilde{h}$, $\tilde{h}'$ are then related by
\begin{equation}
\label{stmetricconfrel}
\tilde{h}' = \theta^4\,\tilde{h} \quad \mbox{with} \quad 
\theta = \left(\frac{1 - v'}{m'}\right)^{-1}\,\vartheta\,\,
\,\frac{1 - v}{m}.
\end{equation}
In this article we investigate the question whether there exist static
vacuum data sets which admit non-trivial conformal rescalings.

The following general transformations map static vacuum
data sets onto such sets:\\
$-$ {\it Trival rescalings}
\begin{equation} \label{trivresc}
\tilde{h} \rightarrow \tilde{h}' = \theta^4\,\tilde{h}, \quad
v \rightarrow v' = v
\quad \mbox{with} \quad \theta = const. > 0. 
\end{equation}
Asymptotic flatness of $\tilde{h}'$ follows with the coordinate
transformation 
$\tilde{x}^a \rightarrow \tilde{x}^{a'} = \theta^2\,\tilde{x}^a$ 
in (\ref{falloff}), which shows
that the mass transforms as $m \rightarrow m' = \theta^2\,m$.
The corresponding conformal factor $\vartheta$ in (\ref{h'-h-rel}) is 
given by
$\vartheta = \frac{1 - v}{m'}\,\,\theta\,
\left(\frac{1 - v}{m}\right)^{-1} = \theta^{-1}$\\
$-$ The transitions
\begin{equation} 
\label{m--m}
(\tilde{S}, \tilde{h}, v) \rightarrow (\tilde{S}, 
\tilde{h}' = v^4\,\tilde{h}, v^{-1}).
\end{equation}
under which the sign of the mass changes, $m \rightarrow m' = - m$, and 
(\ref{h'-h-rel}) holds with $\vartheta = 1$. These
transitions are suggested by the conformal 
static field equations studied below
(cf. the remark following (\ref{Bvomsig})).
Without loss of generality it is therefore sufficient to consider the case
\begin{equation} 
\label{posmass}
m,\, m' > 0.
\end{equation}

If the metric $\tilde{h}$  is {\it conformally flat} with non-vanishing
mass it is necessarily induced by a Schwarzschild solution 
(\cite{friedrich:cargese}). In isotropic coordinates $\tilde{x}^a$ the data
are  
\[
\tilde{h} = 
- \left(1 + \frac{m}{2\,|\tilde{x}|}\right)^4
\delta_{ab}\,d\tilde{x}^a\,d\tilde{x}^b, \quad
v = \frac{1 - \frac{m}{2\,|\tilde{x}|}}{1 + \frac{m}{2\,|\tilde{x}|}}, 
\]
and (\ref{preferred-gauge}) gives in the coordinates
$x^a = |\tilde{x}|^{-2}\,\tilde{x}^a$, which 
are $h$-normal coordinates centered at $i$,
\[
h = - \delta_{ab}\,dx^a\,dx^b, \quad 
\Omega =  \left(\frac{|x|}{1 + \frac{m}{2}\,|x|}\right)^2.
\]

The transition
$\tilde{h}
\rightarrow
\theta^4\,\tilde{h}$,
$v \rightarrow v'$ with
$\theta = \frac{1 + \frac{m'}{2\,|\tilde{x}|}}
{1 + \frac{m}{2\,|\tilde{x}|}}$, 
$v' = \frac{1 - \frac{m'}{2\,|\tilde{x}|}}{1 + \frac{m'}{2\,|\tilde{x}|}}$, 
which maps a Schwarzschild metric $\tilde{h}$ with mass $m$ onto such
a metric with mass $m'$, corresponds to a trivial rescaling combined with a
coordinate transformation. In terms of
$h$ this rescaling is given by (\ref{stmetricconfrel}) with $\vartheta =
1$ and the information on the difference between the solutions is encoded in
the conformal factors $\Omega$ and $\Omega'$. 
The conformally flat case is special
in admitting the $3$-parameter group of {\it special conformal
transformations} 
\[
x^a \rightarrow \frac{x^a + d^a\,x_c\,x^c}
{1 + 2\,d_c\,x^c + d_c\,d^c\,x_e\,x^e},
\quad \quad d^a = const. \in \mathbb{R}^3,
\]
as local, non-trival, conformal maps of $h$ which leave $i$ fixed. 
In terms of the coordinates $\tilde{x}^a$ the map above is given by the
simple translation $\tilde{x}^a \rightarrow \tilde{x}^a + d^a$.
Conversely, the translations $x^a \rightarrow x^a + d^a$ is represented in
terms of the coordinates $\tilde{x}^a$ by a special conformal
transformation which maps a neighbourhood of infinity onto a punctured
neighbourhood of the point $|d|^{-2}\,d^a$.
In the following we shall be interested in static
data which are not conformally flat.

The question posed above has been considered by Beig
(\cite{beig:1991}). He defines a certain quantity $Q$, quadratic in
the Cotton tensor and its covariant derivatives up to second order, and
shows that $h$ cannot admit non-trivial rescalings if $Q \neq 0$ at $i$.
As also pointed out in \cite{beig:1991}, this
condition excludes axi-symmetric static data. Unfortunately it is not
clear what else is excluded.  
We wish to control the situation without imposing conditions and want to
decide whether there do exist static, 
conformally non-flat vacuum data 
that admit non-trivial conformal rescalings. The lemmas
proven in this article imply the following.

\begin{theorem}
\label{exist3parameter-set}
Suppose $(\tilde{S}, \tilde{h}, v)$ are static vacuum data with
ADM mass $m \neq 0$. If they admit a non-trival conformal rescaling, 
then $\tilde{h}$
admits a conformal Killing field $Y$ which is neither homothetic nor a
Killing field for $\tilde{h}$ and, in addition, a Killing field $X$ which
defines an infinitesimal axi-symmetry. The fields
$Y$ and
$X$ commute, are orthogonal to each other, and $Y$ is tangent to the axis
defined by $X$ which passes through the point representing space-like
infinity. Furthermore, if $\tilde{h}$ is not conformally flat it has a
non-vanishing quadrupole moment.

There exists a 3-parameter family of static vacuum data with
$m \neq 0$ which are not conformally flat and which do admit
non-trival conformal rescalings. 
The associated space-time
metrics are of Petrov type $D$ along the axis and of Petrov type $I$ on an
open neighbourhood surrounding the axis.
\end{theorem}
 
This result would provide complete information about the non-conformally
flat static data with non-vanishing mass which admit non-trivial rescalings,
were it not for an unanswered question. As discussed below, it is left open 
whether there exist data 
with non-vanishing quadrupole moment  which admit non-trivial rescalings with
$d \vartheta = 0$ at $i$.

Apart from this omission the result above represents a 3-dimensional
analogue of Brinkmann's theorem. Brinkmann studied the solutions of
Einstein's vacuum field equations in 4 space-time dimensions which admit
non-trivial conformal rescalings that yield new vacuum fields (asymptotic behaviour, however,  played no role   
in these studies). He found
them to be given by the solutions which have later been named {\it vacuum
pp-waves}  (\cite{brinkmann},
\cite{stephani:et:al}).  These solutions are of Petrov type $N$. We note
that the rescalings of the static data considered in this article do not
extend to conformal rescalings of the correponding static vacuum
space-times. 

Details of the static data whose existence has been shown here
will be discussed elsewhere.

\section{Conformal static vacuum field equations}

The derivation of the following equations has
been discussed (in terms of $h$ and $\zeta = \rho/\mu$) in detail in
\cite{friedrich:statconv}. Using the conformal metric
$h$ defined by (\ref {tildeh-h}), (\ref{preferred-gauge}) and the function
\begin{equation}
\label{Avomsig}
\rho = \left(\frac{1 - v}{1 + v}\right)^2,
\end{equation}
one obtains the static vacuum equations (\ref{statEinstvac}) in the form
\begin{equation}
\label{n1fequ}
0 = \Sigma_{ab}[h, \mu] \equiv D_a D_b \rho -
s\,h_{ab} + \rho\,(1 - \rho)\,s_{ab}
\quad \mbox{with} \quad 
s = \frac{1}{3}\,\Delta_h \rho,
\end{equation}
\begin{equation}
\label{srtsigmaequ}
(\Delta_h - \frac{1}{8}\,R[h])\,
[\frac{1}{\sqrt{\rho}}] = \frac{4\,\pi}{\sqrt{\mu}}\,\delta_i.
\end{equation} 
Here $D$ denotes the covariant derivative defined by $h$, 
and the tensor
\[
s_{ab} = R_{ab}[h],
\]
is the Ricci tensor of $h$. It is trace
free because the Ricci scalar of $h$  satisfies in the conformal gauge defined by 
(\ref{preferred-gauge})
\begin{equation}
\label{gaugechar}
R[h] = 0.
\end{equation}
We note that the tensor $\frac{m}{2}\,s_{ab}(i)$ at space-like infinity
represents the quadrupole moment of the static metric $\tilde{h}$. 
We set 
\begin{equation}
\label{mudef}
\mu = \frac{m^2}{4},
\end{equation}
and denote by $\delta_i$ the Dirac measure
with weight one and support $\{i\}$, given by the standard Dirac measure
$\delta_0$ in $h$-normal coordinates $x^a$ centered at $i$. In
such coordinates the fields $h$ and $\rho$ are real analytic and
satisfy 
\begin{equation}
\label{sigmaval}
\rho = 0,\,\,\,\,\,\,D_a\rho = 0,\,\,\,\,\,\,
D_a D_b \rho = - 2\,\mu\,h_{ab}
\quad \mbox{at} \quad i.
\end{equation}

The function $\rho$ can be characterized as the unique real analytic solution of
(\ref{srtsigmaequ}),  (\ref{sigmaval}) or as the unique real analytic function satisfying 
(\ref{sigmaval}) and
\begin{equation}
\label{an2fequ}
2\,\rho\,s = D_a\rho\,D^a\rho,
\end{equation}
which is a rewrite of (\ref{srtsigmaequ}).
This equation can be shown (\cite{friedrich:cargese}) to be in fact a consequence
of (\ref{n1fequ}) and (\ref{sigmaval}) so that the essential information on the
static field equations is encoded in (\ref{n1fequ}), (\ref{sigmaval}).

In going from $h$, $\rho$ back to $\tilde{h}$ and $v$, one has to choose the sign
of the square root. In the case of positive mass the correct formulas are
\begin{equation}
\label{Bvomsig}
\Omega = \frac{\rho}{\mu\,(1 + \sqrt{\rho})^2},
\,\,\,\,\,\,\,\,\,
v = \frac{1 - \sqrt{\rho}}{1 + \sqrt{\rho}}.
\end{equation}
Replacing here the square root by its negative
amounts to the transition (\ref{m--m}).

\vspace{.3cm}

The system (\ref{n1fequ}) implies the integrability conditions
\begin{equation}
\label{B1intco}
0 = \frac{1}{2}\,D^e\,\Sigma_{ea} =
D_a\,s + (1 - \rho)\,s_{ab}\,D^b\,\rho,
\end{equation}
and 
\begin{equation}
\label{Bcott}
0 = 
\frac{1}{\rho}\,(D_{[c}\, \Sigma_{a]b} + 
\frac{1}{2}\,D^e\,\Sigma_{e[c}\,h_{a]b})
= (1 - \rho)\,D_{[c}s_{a]b} 
- 2\,D_{[c}\rho\,s_{a]b} 
- D^d\,\rho\,s_{d[c}\,h_{a]b},
\end{equation}
which extends by analyticity also to $i$. 

\vspace{.3cm}

While the static vacuum data are subject to a
rescaling (\ref{h'-h-rel}), the
transformation of the potential was left unspecified. As pointed out above,
the potential, represented by the
function $\rho$, is determined in the asymptotically flat case uniquely by
$\mu$ and $h$. The new potential should thus be given in terms of $\rho$,
the conformal factor $\vartheta$, and the new mass term $\mu'$.
\begin{lemma}
\label{rescfundsol}
Assume that $(S, h, \rho)$ is derived from a static vacuum data set with mass
$m > 0$ as discussed above and  $\vartheta > 0$ is a $C^2$ function so that
$h' = \vartheta^{4}\,h$ is the conformal metric associated with static
vacuum data $\tilde{h}'$, $v'$ on $S \setminus \{i\}$ with mass $m' >
0$. Set 
\begin{equation}
\label{nufindef}
\nu = \frac{m^2}{m'^2\,\vartheta(i)^4}.
\end{equation} 
Then the function  
\begin{equation}
\label{fsolcontraf}
\rho' = \frac{1}{\nu}\left(\frac{\vartheta}{\vartheta(i)}\right)^2\,\rho,
\end{equation}
satisfies 
\begin{equation}
(\Delta_{h'} - \frac{1}{8}\,R[h'])[\frac{1}{\sqrt{\rho'}}] = 
\frac{4\,\pi}{\sqrt{\mu'}}\,\delta_i,
\end{equation}
and relations (\ref{sigmaval}) with $\rho$, $\mu$, $D$ replaced by $\rho'$,
$\mu'$ and the covariant derivative operator $D'$ defined by $h'$.
The function $\rho'$ agrees with the one given by (\ref{Avomsig}) with $v$
replaced by $v'$.
\end{lemma}

\begin{proof}
Let $x^a$ resp. $x^{a'}$ denote $h$- resp. $h'$-normal coordinates
centered at $i$. A calculation then shows that the system $x^{a'}$ satisfies,
possibly after a rotation around the origin, the relation
\begin{equation}
\label{ncoordtrafo}
x^{a'} 
= \delta^{a'}\,_a \left\{\vartheta(i)^2\,x^a 
- \vartheta(i)\,(\delta_{bc}\,\delta^{ad}
- 2\,\delta^a\,_b\,\delta^d\,_c)\,\vartheta_{,d}(i)\,x^a\,x^c \right\} 
+ O(|x^{a}|^3).
\end{equation}
Writing the transformation $x^a = x^a(x^{c'})$ shortly $x
= f(x')$, the transformation behaviour of Dirac distributions under
coordinate transformations implies with the relation above 
$f^*\,\delta_0 = |\det\,\frac{\partial f}{\partial x'}(i)|^{-1}\,\delta_{0'}
= \vartheta(i)^6\,\delta_{0'}$ (\cite{hoermander:I}). With the conformal
covariance of the conformal Laplacian and (\ref{srtsigmaequ}) it thus follows 
\[
f^*\left((\Delta_{h'} -
\frac{1}{8}\,R_{h'})[\frac{1}{\sqrt{\rho'}}]\right) = 
\sqrt{\frac{\mu}{\mu'}}\,\vartheta(i)^{-1}\,
f^*\left((\Delta_{h'} -
\frac{1}{8}\,R_{h'})[\frac{1}{\vartheta\,\sqrt{\rho}}]\right)
\]
\[
= \sqrt{\frac{\mu}{\mu'}}\,\vartheta(i)^{-1}\,
f^*\left(\vartheta^{-5}\,(\Delta_{h}
- \frac{1}{8}\,R_{h})[\frac{1}{\sqrt{\rho}}]\right)
= \frac{4\,\pi}{\sqrt{\mu'}}\,\delta_{0'}.
\]
The relations (\ref{sigmaval}) are verified by a direct calculation and the last
statement follows by the uniqueness property pointed out above.
\end{proof}

\vspace{.3cm}

The quantity $\nu$ is left unchanged under trivial
conformal rescalings. It has the following meaning. The conformal factor
$\theta$ has in the coordinates of (\ref{falloff}) in general an expansion
\[
\theta = \frac{1}{\vartheta(i)}\left( 1 + \frac{a}{|\tilde{x}|}
+ O\left(\frac{1}{|\tilde{x}|^2}\right) \right),
\]
with some coefficients $\vartheta(i) > 0$ and $a$. Rescaling the metric 
(\ref{falloff}) with $\theta$ one finds that the rescaled metric $\theta^4
\tilde{h}$ acquires the mass $m' = \vartheta(i)^{-2}(2\,a + m)$. Using
(\ref{Bvomsig}), (\ref{fsolcontraf}) in the expression (\ref{stmetricconfrel}) for
$\theta$, expanding $\rho$ using (\ref{0-D-DD_Omega-at-i}), and comparing
with the expression for $\theta$ above, we find
$2\,a = m\left(\nu^{-1/2} - 1\right)$ and thus
again (\ref{nufindef}). 
A change of mass is thus generated purely by a trivial rescaling 
if $\nu = 1$ but is partly due to an independent contribution if $\nu \neq
1$.
We finally note the expressions
\begin{equation}
\label{theta-v'}
\theta = \frac{\sqrt{\nu}\,\vartheta(i) +
\sqrt{\rho}\,\vartheta}{(1 +
\sqrt{\rho})\,\sqrt{\nu}\,\vartheta(i)^2},
\quad \quad
v' = 
\frac{\sqrt{\nu}\,\vartheta(i) - \vartheta 
+ (\sqrt{\nu}\,\vartheta(i) + \vartheta)\,v}
{\sqrt{\nu}\,\vartheta(i) + \vartheta 
+ (\sqrt{\nu}\,\vartheta(i) - \vartheta)\,v},
\end{equation}
for the conformal factor and the transformed potential. It follows that 
$\theta = 1$ if and only if $\nu = 1$ and $\vartheta = 1$, while
$v' = v$ is equivalent to $(\sqrt{\nu}\,\vartheta(i) - \vartheta)\,v^2
= (\sqrt{\nu}\,\vartheta(i) - \vartheta)$. This can hold only
if $v = 1$, which implies that $\tilde{h}_{ab}$ is flat, or if 
$\vartheta = \sqrt{\nu}\,\vartheta(i)$, which implies that $\nu = 1$,
$\vartheta = \vartheta(i)$, and $\theta = \vartheta(i)^{-1}$.

\section{The equations for the rescaling factor}

It will be convenient to replace  $\vartheta$ by
$\gamma\,\vartheta$ and assume  
\begin{equation}
\label{Anconfresc}
\vartheta(i) = 1, \quad \quad \gamma = const. > 0,
\quad \quad \nu = \frac{\mu}{\mu'\,\gamma^4},
\end{equation}
\begin{equation}
\label{Bnconfresc}
h' = \gamma^4\,\vartheta^{4}\,h,
\,\,\,\,\,\,\,\,\,\,\,\,
\rho' = \frac{1}{\nu}\,\vartheta^{2}\,\rho.
\end{equation}

To derive conditions on the scaling factors we express
\[
\Sigma_{ab}[h', \mu'] \equiv D'_a D'_b \rho' - s'\,h'_{ab}
+ \rho' (1 - \rho')\,s'_{ab},
\]
in terms of $h$ and $\rho$. 
With the general rescaling laws in $3$ dimensions 
\begin{equation}
\label{rictencontraf}
R_{ab}[\vartheta^4\,h] = R_{ab}[h] - 2\,\vartheta^{-1}\,D_aD_b\vartheta 
+ 6\,\vartheta^{-2}\,D_a\vartheta\,D_b\vartheta
- 2\,h_{ab}\left\{\vartheta^{-1}\,D_cD^c\vartheta 
+ \vartheta^{-2}\,D_c\vartheta\,D^c\vartheta\right\},
\end{equation}
\begin{equation}
\label{ricscalcontraf}
 \frac{1}{8}\,R[ \vartheta^4\,h]\,\vartheta^5 =
- \Delta_h\vartheta +  \frac{1}{8}\,R[h]\,\vartheta,
\end{equation}
where the right hand sides are expressed in terms of quantities derived from
$h$, a direct calculation gives
\[
\Sigma_{ab}[h', \mu'] = 
\gamma^{4}\,\vartheta^2\,\frac{\mu'}{\mu}\,\Sigma_{ab}[h,
\mu] +
\frac{8}{3}\,\gamma^{4}\,\frac{\mu'}{\mu}\,\left(\frac{1}{\nu}\, \vartheta^3
\rho^2 - \vartheta\,\rho\right)\, \Delta_h \vartheta\, h_{ab}
\]
\[
- \gamma^{4}\,\frac{\mu'}{\mu\,\nu^2}\,\vartheta^6\,
\rho^2\left( D_a D_b u - \frac{1}{3}\,\Delta_h u \,h_{ab}
+ u\,(1 - u)\,s_{ab}\right),
\]
where we set 
\begin{equation}
\label{udef}
u = \nu\,\vartheta^{-2},
\end{equation}
and used the resulting relation
\begin{equation}
\label{Laprewrite}
2\,u\,\Delta_h u - 3\,D_c u D^c u = - 
4\,\nu^2\,\vartheta^{-5}\,\Delta_h \vartheta.
\end{equation}
Equation (\ref{ricscalcontraf}) implies with $R[h] = 0$ and
$R[h'] = 0$ 
\begin{equation}
\label{Lapvardelvan}
\Delta_h\vartheta = 0.
\end{equation}
From these relations we read off the following basic condition.
\begin{lemma}
\label{basicscalecond}
Suppose $h$ and $\rho$ satisfy
(\ref{n1fequ}) and (\ref{sigmaval}) with some constant $\mu > 0$.
If $\mu', \gamma > 0$ are constants and $\vartheta$ a positive function with
$\vartheta(i) = 1$, then
$h' = \gamma^4\,\vartheta^{4}\,h$ and
$\rho' = \frac{1}{\nu}\,\vartheta^{2}\,\rho$ with 
$\nu = \frac{\mu}{\mu'\,\gamma^4}$
satisfy $\Sigma_{ab}[h', \mu'] = 0$ if and only if 
$u = \nu\,\,\vartheta^{-2}$ satisfies the overdetermined system
\begin{equation}
\label{Lapvan}
0 = \Pi[h, u] \equiv 2\,u\,t - D_c u D^c u 
\quad \mbox{with} \quad t = \frac{1}{3}\,\Delta_h u,
\end{equation}
\begin{equation}
\label{DDvan}
0 = \Pi_{ab}[h, u] \equiv D_a D_b u - t\,h_{ab}
+ u\,(1 - u)\,s_{ab}.
\end{equation}
Moreover, $u$ must satisfy the initial condition
\begin{equation}
\label{uati}
u(i) = \nu.
\end{equation}
\end{lemma}

Using the Bianchi identity and the decomposition
$R_{dbca} = 2\,(h_{d[c}\,s_{a]b} - h_{b[c}\,s_{a]d})$,
which holds because $R[h] = 0$ and $dim(S) = 3$, one gets from (\ref{DDvan})
the integrability conditions
\begin{equation}
\label{contrid}
0 = \frac{1}{2}\,D^c \Pi_{ca} = 
D_a \,t + (1 - u)\,D^cu\,s_{ca}, 
\end{equation}
\begin{equation}
\label{biguid}
0 = \frac{1}{u}\left(D_{[c} \Pi_{a]b} 
+ \frac{1}{2}\,D^d \Pi_{d[c}\,h_{a]b}\right)
= (1 - u)\,D_{[c}\,s_{a]b}
- 2\,D_{[c}u\,s_{a]b} - D^du\,s_{d[c}\,h_{a]b}.
\end{equation}
The identity
\begin{equation}
\label{Piint}
D_a \Pi = u\,D^c \Pi_{ca} - 2\,D^cu \,\Pi_{ca},
\end{equation}
implies that (\ref{Lapvan}) will be satisfied by a solution $u$ of 
(\ref{DDvan}) if $\Pi(i) = 0$ i.e. if
\begin{equation}
\label{uDutindata}
u(i) = \nu > 0, \quad
2\,\nu\,t(i) = c_a\,c^a
\quad \mbox{with} \quad c_a = D_a u(i).
\end{equation}

\vspace{.2cm}

Let $x(\tau)$ denote a geodesic through $i$ with unit tangent vector
$\dot{x}(\tau)$. Transvecting equations (\ref{DDvan}),
(\ref{contrid}) suitably with
$\dot{x}$, one obtains a system of ODE's for $u$, $D_au$, $t$ along this 
curve, which shows that a solution $u$ of (\ref{DDvan}), (\ref{contrid}), if it
exists, must be analytic and uniquely determined by the data 
\[
u(i) = \nu > 0, \quad
\dot{u}(i) = \dot{x}^ac_a, \quad 
t(i) = \frac{1}{2\,\nu}\,c_a\,c^a.
\]
The function $u$ so obtained will in
general not satisfy the complete system
(\ref{DDvan}). It will be shown that the existence of
non-trivial solutions to equation (\ref{DDvan}) imposes strong restrictions
on the metric $h$. 
Because of the factor $1 - u$ in equations 
(\ref{DDvan}), (\ref{contrid}) it follows immediately that 
\begin{equation}
\label{nu1ca0}
u \equiv 1 \quad \mbox{if} \quad \nu = 1, \quad c_a = 0.
\end{equation}
The following result will be useful later.
\begin{lemma}
\label{genurhorel}
Let $u$ be a solution to (\ref{DDvan}),
(\ref{uDutindata}) on a 
neighbourhood of $i$ on which $\rho < a$ for some $a> 0$.   If
$u = F(\rho)$
with some function $F \in C^2([0, a[)$, then $u$ is the trivial solution 
$u = \nu$. Moreover, $h$ is flat unless $\nu = 1$. 
\end{lemma}

\begin{proof} Observing (\ref{an2fequ}) one gets 
\[
2\,u\,t - D_c u D^c u =
\frac{1}{3}\,D_a\rho\,D^a\rho\,
\left(\frac{3}{\rho}\,F\,F' + 2\,F\,F'' - 3\,(F')^{2}\right),
\]
so that (\ref{Lapvan}) is equivalent to the ODE
\[
\frac{3}{\rho}\,F\,F' + 2\,F\,F'' - 3\,(F')^{2} = 0.
\]
If $F' = 0$ at a point where $F > 0$ it follows that 
$u = F = const. = \nu$ which implies with
(\ref{DDvan}) that $s_{ab} = 0$ unless $\nu = 1$. 
Near a point where $F > 0$ and $F' \neq 0$ the equation above implies
$0 = \frac{3}{\rho} + 2\,\frac{F''}{F'} - 3\,\frac{F'}{F}
=   \left(\log \frac{\rho^3\,(F')^2}{F^3}\right)'$
whence 
$F = \frac{\rho}{(b + d\,\sqrt{\rho})^2}$
with $b, d = const. > 0$, $b \neq 0$.
But then $u(p) = F(\rho(p)) \rightarrow 0 \neq \nu$ as $p \rightarrow i$,
which contradicts our assumptions.
\end{proof}

\section{Implications for $h$}

Let $\tilde{h}$, $v$ be static data and $h$, $\rho$ the associated conformal
fields. We shall discuss now properties of $\tilde{h}$ and $h$
which are implied by the existence of a
{\it non-trivial} solution $u$ to (\ref{Lapvan}), (\ref{DDvan}), (\ref{uati}). 
Assume that 
\begin{equation}
\label{Udef} 
U \,\,\,\mbox{{\it is an $i$-centered, convex $h$-normal nbhd so
that}}\,\,\,0 < \rho < 1, \,\,D_a\rho \neq 0\,\,
\mbox{on}\,\,U \setminus \{i\},
\end{equation}
and set with a  given function $u$
\begin{equation}
\label{waform}
\rho_a = D_a\rho, \quad 
u_a = D_au, \quad
w = \frac{1 - u}{1 - \rho},
\quad
w_a = D_aw,
\quad 
U^* = \{p \in U|\,\,w_a(p) \neq 0\},
\end{equation}
so that $i \in U^*$ if and only if $u_a(i) \neq 0$.
We recall that $h$ is
conformally non-flat if and only if the set $\{p \in U|\,\,s_{ab}(p) \neq
0\}$ is dense in $U$ (\cite{friedrich:cargese}).
\begin{lemma}
\label{ricciexpr}
Let $h$, $\rho$ denote a solution to
(\ref{n1fequ}), (\ref{sigmaval}) which is not conformally flat. 
Suppose $u$ is a non-constant, positive solution to (\ref{Lapvan}),
(\ref{DDvan}), (\ref{uati}) on a set $U$ satisfying 
(\ref{Udef}) and define $w$ as in (\ref{waform}).
Then the set $U^*$ is dense in $U$ and there exists a  smooth function $\beta$ on
$U^*$ so that
\begin{equation}
\label{2ndbigindcond}
s_{ab}
= \beta\,(w_{a}\,w_b - \frac{1}{3}\,h_{ab}\,w_{c}\,w^c).
\end{equation}
If $V \subset  U^*$ is a connected, simply connected neighbourhood of a point
$p \in U^*$, there exist a constant $\beta_* \neq 0$ and a function $H =
H(w)$ defined on $V$ with
$H(w(p)) = 0$ so that the Ricci tensor has on $V$ the representation
\begin{equation}
\label{Ricexpr}
s_{ab} = \frac{\beta_*}{1 - \rho}\,e^H\,
(w_a\,w_b - \frac{1}{3}\,h_{ab}\,w_c\,w^c).
\end{equation}
If  $u_a(i) \neq 0$ we can choose $p = i$.
\end{lemma}

\begin{proof} If $w_a$ vanished on an open subset of $U$, 
$w$ would be constant on
$U$ because $u$ and $\rho$ are analytic.
It would follow that $u = \nu + \rho\,(1 - \nu)$ whence $u = const.$
by lemma \ref{genurhorel}, in conflict with our assumptions.
It follows that $U^*$ is dense in $U$. Because
\[
\frac{1}{u\,(1 - \rho)}\left(D_{[c} \Pi_{a]b} 
+ \frac{1}{2}\,D^d \Pi_{d[c}\,h_{a]b}\right)
=
\]
\[
\frac{1}{u\,(1 - \rho)}\left(D_{[c} \Pi_{a]b} 
+ \frac{1}{2}\,D^d \Pi_{d[c}\,h_{a]b}\right)
- \frac{1 - u}{\rho\,(1 - \rho)^2}\,(D_{[c}\, \Sigma^*_{a]b} + 
\frac{1}{2}\,D^e\,\Sigma^*_{e[c}\,h_{a]b})
\]
\[
= 2\,w_{[c}\,s_{a]b} + w^d\,s_{d[c}\,h_{a]b},
\]
equation (\ref{biguid}) holds on $U^*$ if and only if
\begin{equation}
\label{firstindcond}
2\,w_{[c}\,s_{a]b} + w^d\,s_{d[c}\,h_{a]b} = 0.
\end{equation}
Contraction with $2\,w^c$ gives   
\[
2\,w_{c}\,w^c\,s_{ab} - 2\,w_{a}\,s_{cb}\,w^c 
+ w^d\,s_{dc}\,w^c\,h_{ab} - w^d\,s_{da}\,w_{b} = 0.
\]
The antisymmetric part of this equation reads $w^d\,s_{d[a}\,w_{b]} = 0$,
which implies on $U^*$ 
\begin{equation}
\label{bigeigenvalue}
w^d\,s_{da} = \alpha\,w_a,
\end{equation}
with some function $\alpha$. Using this in the equation above, we obtain 
(\ref{2ndbigindcond}), which satisfies (\ref{firstindcond}) without restriction
on $\beta$.

With (\ref{n1fequ}) and (\ref{DDvan}) one obtains 
\begin{equation}
\label{DDw}
D_a w_b = J\,h_{ab}
+ (1 - u)\,(1 - w)\,s_{ab}
+ \frac{1}{1 - \rho}(w_a \rho_b + w_b \rho_a)
\quad \mbox{with} \quad
J = \frac{w\,s - t}{1 - \rho}.
\end{equation} 
It follows with (\ref{2ndbigindcond})
\[
D_a w^a = 3\,J + \frac{2}{1 - \rho}\,\rho_a\,w^a,
\]
\[
w^a\,D_b\,w_a = J\,w_b 
+ \frac{2}{3}\,(1 - u)\,(1 - w)\,\beta\,w_c\,w^c\,w_b
+ \frac{1}{1 - \rho}(w_a w^a \rho_b + w_b \,\rho_a w^a),
\]
\[
w^b\,w^a\,D_b\,w_a = \left\{J 
+ \frac{2}{3}\,(1 - u)\,(1 - w)\,\beta\,w_c\,w^c
+ \frac{2}{1 - \rho}\,\rho_a\, w^a \right\}w_c\,w^c.
\]
On $U^*$ the Bianchi identity and the gauge condition $R[h] =
0$ imply with (\ref{2ndbigindcond}) 
\[
0 = D^as_{ab} = w^aD_a\beta\,\,w_b - \frac{1}{3}\,D_b\beta\,\,w_c w^c
+ \beta\,(D_aw^a\,w_b  
+ \frac{1}{3}\,w^c D_bw_c).
\]
After contraction with $w^b$ this can be solved for $w^aD_a\beta$. Inserting
the resulting expression again  into the equation, gives
\begin{equation}
\label{dbeta}
w_cw^c\,\,D_a\beta =
\beta\,\left\{
- \frac{3}{2}\left(D_cw^c + \frac{w^bw^cD_b w_c}{w_cw^c}\right)w_a
+w^cD_aw_c \right\}
\end{equation}
\[
= \beta\,\left\{ - 5\,J\,w_a - \frac{5}{1 - \rho}\,\rho_c\, w^c\,w_a
-  \frac{1}{3}\,(1 - u)\,(1 - w)\,\beta\,w_c\,w^c\,w_a
+  \frac{1}{1 - \rho}\,w_c\,w^c\,\rho_a \right\}.
\]
and thus finally 
\begin{equation}
\label{betaequ}
D_a\beta = \beta \left\{
\frac{1}{1 - \rho}\,\rho_a + K\,w_a\right\},
\end{equation}
with 
\[
w_cw^c\,\,K = - 5\,J
- \frac{5}{1 - \rho}\,\rho_c\,w^c
- \frac{1}{3}\,(1 -  u)\,(1 - w)\,w_cw^c\,\beta.
\]

The relation $\beta(p) = 0$ would imply with (\ref{betaequ}) that
$\beta = 0$, whence $s_{ab} = 0$ on a neighbourhood of $p$ in
$U^*$  and thus $s_{ab} = 0$ on $U$ by analyticity, contradicting conformal
non-flatness. 
Thus $\beta(p) \neq 0$, whence $\beta \neq 0$ on $V$, and equation
(\ref{betaequ}) can be written there in the form
$D_a \log|(1 - \rho)\,\beta| = K\,D_aw$.
This implies that $D_{[a}K\,D_{b]}w = 0$, $K$ can be written as a function of
$w$ on $V$, and there exists a function $H = H(w)$ 
with $H(w(p)) = 0$ so that
$D_a (\log|(1 - \rho)\,\beta| - H) = 0$,
whence, with $\beta_* = (1 - \rho)\,\beta|_p \neq 0$, 
\begin{equation}
\label{betaexpr}
\beta = \frac{\beta_*}{1 - \rho}\,e^H
\quad \mbox{on} \quad V.
\end{equation} 
\end{proof}

\vspace{.3cm}

To state the following result, we note that lemma \ref{genurhorel} implies 
under the assumptions of  lemma \ref{ricciexpr} that
$\epsilon^{abc}\,u_b\,\rho_c \neq 0$ on a dense open subset of $U$.
\begin{lemma}
\label{existKilling}
Assume the notation and the assumptions of lemma \ref{ricciexpr}. If
$V$ is chosen so that $\epsilon^{abc}\,u_b\,\rho_c \neq 0$ on $V$, there exists a
function $l = l(w)$ on $V$ so that 
\begin{equation}
\label{Killingfield}
X^a = l\,\epsilon^{abc}\,u_b\,\rho_c,
\end{equation}
defines a non-trivial Killing field for
$h$ on $V$. 
It extends to an analytic, hypersurface orthogonal Killing field $X$
on $U$ which satisfies the relations
\begin{equation}
\label{XrhoXuXw}
X^a \rho_a = 0, \quad X^a u_a = 0,
\end{equation}
and is a Killing field for $\tilde{h}$ on $U \setminus \{i\}$.
Independent of the choice of V it is determined uniquely up to a
constant real factor.
It vanishes at $i$ and defines an infinitesimal axi-symmetry with closed
integral curves near $i$. If $u_a(i) \neq 0$ its axis is given by the 
$h$-geodesic $\gamma(\tau)$ with tangent vector $u^a(i)$ at $i$.
\end{lemma}

\begin{proof}
With the expression (\ref{Killingfield}) a direct
calculation using (\ref{n1fequ}), (\ref{DDvan}), (\ref{2ndbigindcond}),
(\ref{betaexpr}) gives on $V$
\[
D_aX_b = \epsilon_{ab}\,^c\,\eta_c
+ l^{-1}\left\{D_bl + l\,\beta_*\,e^H\,(w - w^2)\,w_b
\right\}X_a,
\]
where
\begin{equation}
\label{Xetadef}
\eta_c = \frac{1}{3}\,l\left\{\left(\Delta_h\rho 
+ \rho\,(1 - \rho)\,\beta\,w_d\,w^d\right)\,u_c
- \left(\Delta_hu
+ u\,(1 - u)\,\beta\,w_d\,w^d\right)\,\rho_c\right\}.
\end{equation}
With the choice
 \begin{equation}
\label{lchoice}
l = l_*\,e^{L(w)}
\quad \mbox{where $L$ satisfies} \quad L'(w) = - \beta_*\,e^H\,(w - w^2)
\quad \mbox{and} \quad l_* = const. > 0,
\end{equation}
this implies
\begin{equation}
\label{DX}
D_aX_b = \epsilon_{ab}\,^c\,\eta_c.
\end{equation}
It follows 
\begin{equation}
\label{hypsurforthoKilling}
D_{(a}X_{b)} = 0, \quad \epsilon^{abc}\,X_a D_b\,X_c = 0.
\end{equation}
Equations (\ref{XrhoXuXw}) are an immediate consequence of the definition of
$X$. Since then $X^aD_a\Omega = 0$ by (\ref{Bvomsig}), $X$ is also a Killing
field for the metric $\tilde{h}$.

To construct the extension of $X$ to $U$ we use the integrability
condition for the Killing equation, 
\begin{equation}
\label{Killingintcond}
D_aD_bX_c = X_d\,R^d\,_{abc},
\end{equation}
which is satisfied by $X$ on $V$ and to be
satisfied on $U$. Fix $p \in V$. The geodesics $\gamma(\tau)$ with
$\gamma(0) = p$ cover the convex normal neighbourhood $U$. The
ODE's
\begin{equation}
\label{Killingintcondequ}
D^2_{\dot{\gamma}}X_c = X_d\,R^d\,_{abc}\,\dot{\gamma}^a\,\dot{\gamma}^b,
\end{equation}
along these geodesics determine a unique analytic extension of
$X$ to $U$. The Killing
and the hypersurface orthogonality conditions (\ref{hypsurforthoKilling}) as well
as the relations (\ref{XrhoXuXw}) extend to
$U$ by analyticity. As above it follows that the extended field is also a Killing
field for $\tilde{h}$.

The relation $X^{[a}\,\epsilon^{b]cd}\,u_c\,\rho_d = 0$,
which holds on $V$, extends to $U$ by analyticity so that
\begin{equation}
\label{Xproexpl}
X^a \sim \,\epsilon^{acd}  u_c\,\rho_d \quad \mbox{on} \quad U.
\end{equation}
The factor $l$ which relates these fields on $V$ needs a priori not be bounded
on $U$ but up to a constant factor the Killing field $X^a$ is 
determined uniquely.

The relation  ${\cal L}_X\rho = 0$ on $U$ implies 
\begin{equation}
\label{DLcomm}
0 = D_a\,{\cal L}_X\rho = {\cal L}_X D_a\rho
= X^bD_bD_a\rho + D_aX^bD_b\rho,
\end{equation}
which reduces at $i$ to 
\[
0 = - 2\,\mu\,X_a(i).
\] 
It follows that $D_aX_b \neq 0$ at $i$, otherwise $X$ would vanish
identically. Any geodesic through $i$ is
mapped by the flow of
$X$ onto such geodesic because $i$ is a fixed point of the flow.
Because $dim(S) = 3$ and $D_aX_b$ is
anti-symmetric there exists a tangent vector $t^a \neq 0$ at $i$
with $t^a D_aX_b = 0$. This vector is
invariant under the flow of $X$. The
geodesic $\gamma(\tau)$ satisfying $\gamma(0) = i$ and $\dot{\gamma}(0) = t$
is thus pointwise invariant under the flow so that $X|_{\gamma(\tau)} = 0$
and the points of $\gamma(\tau)$ represent the axis of $X$. Because the
flow of
$X$ preserves orthogonality and maps geodesics onto such, it maps any
geodesic orthogonal to
$\gamma$ onto another such geodesic. Since it preserves affine
parameters it follows that the flow lines of $X$ are closed near $i$.

If $u_a(i) \neq 0$ the function $\beta$ given by  (\ref{betaexpr}) and thus the
function $l$ can be given on a neighbourhood of $i$ so that $l = l_*$ and 
$e^H = 1$ at $i$
and the expressions (\ref{Killingfield}) and (\ref{DX}) with (\ref{Xetadef}) can be
assumed to hold on this neighbourhood.
While then $X^a = 0$ at $i$, we have
$\eta_c(i) = - 2\,l_*\,u_c$,
so that $D_aX_b|_i = - 2\,l_*\,\epsilon_{ab}\,^c\,u_c(i) \neq 0$
and $u^a\, D_aX_b|_i = 0$. It follows that the axis is given by the geodesic 
with tangent vector $u^a$ at $i$. 
\end{proof}

\vspace{.3cm}

If $u_a(i) \neq 0$ we normalize $X$ by setting
\begin{equation} 
\label{normalizingX}
l_* = \frac{1}{2\,c}
\quad \mbox{with} \quad c = \sqrt{ - u_cu^c}|_i\,> 0
\quad \mbox{so that} \quad \eta_a \eta^a|_i = - 1. 
\end{equation}
It follows with (\ref{Killingintcond}) that
$\eta_a \eta^a = D_aX_b\,D^aX^b/2 = const. = - 1$ along the 
the axis, where $X$ vanishes. The flow of $X$ induces then rotations of the
tangent space
$T_iS$ with period $2\,\pi$ and not smaller. If $\phi$ denotes the natural
parameter on the integral curves of $X$ which vanishes on a hypersurface
orthogonal to $X$ which approaches the axis from one side, it defines a
coordinate with  hypersurfaces $\{\phi = const.\} \perp X$ and the integral
curves of
$X$ close exactly  if $\phi \in [0, 2\,\pi[$.

\begin{lemma}
\label{existsconf Killing}
Assume the notation and the assumptions of lemmas \ref{ricciexpr},
\ref{existKilling}. 
The field 
\begin{equation}
\label{conKillingexists}
Y^a = f\,w^a  \quad \mbox{with} \quad f = \frac{l}{l_*}\,(1 - \rho)^2,
\end{equation}
satisfies on $V$ 
\begin{equation}
\label{DK}
D_a Y_b = \omega\,h_{ab} + \frac{1}{l_*}\,\epsilon_{ab}\,^c\,X_c
\quad \mbox{with} \quad 
\omega = f\,(J - \frac{1}{3}(w - w^2)\,w_c w^c\,\beta_*\,e^H).
\end{equation}
It extends to an analytic, hypersurface orthogonal field $Y$
on $U$ satisfying the conformal Killing equation
\begin{equation}
\label{KconKil}
D_a Y_b + D_b Y_a = 2\,\omega\,h_{ab},
\end{equation}
and the relations 
\begin{equation}
\label{XKcommute}
X_a Y^a = 0, \quad\,\, [X, Y] = 0.
\end{equation}
It is tangential to the axis through $i$ defined by $X$, where it does not
vanish. 
It is a conformal Killing field but neither homothetic nor a
Killing field for $\tilde{h}$. 
Independent of the choice of $V$ the extended field is unique up to a
non-vanishing constant real factor.
\end{lemma}

\begin{proof}
With (\ref{DDw}) we get on $V$
\[
D_a Y_b = e^L\,((1 - \rho)^2\,L'\,w_a - 2\,(1 - \rho)\,\rho_a)\,w_b
\]
\[
+ f\left(J\,h_{ab}
+ (1 - u)\,(1 - w)\,\beta(w_aw_b - \frac{1}{3}\,h_{ab}\,w_c\,w^c)
+ \frac{1}{1 - \rho}(w_a \rho_b + w_b \rho_a)\right)
\]
\[
= f \left(
J - \frac{1}{3}\,(w - w^2)\,w_c w^c\,\beta_*\,e^H \right)h_{ab}
+ (1 - \rho)\,e^L\,(w_a \rho_b - w_b \rho_a),
\]
\[
=  \omega\,h_{ab} - e^L\,(u_a \rho_b - u_b \rho_a)
=  \omega\,h_{ab} + \frac{1}{l_*}\,\epsilon_{ab}\,^c\,X_c,
\]
which implies (\ref{DK}), (\ref{KconKil}) on $V$.
Hypersurface orthogonality and orthogonality to $X$
follow immediately from the definitions of $X$ and $Y$.  
The second of relations (\ref{XKcommute})
follows by a direct calculation from
\[
[X, Y]_a = X^cD_c Y_a - Y^cD_cX_a = \omega\,X_a - 
f\,w^c\,\epsilon_{cad}\,\eta^d 
\]
\[
= \omega\,X_a - f\,\frac{1}{3}\left[\left(\Delta_h\rho 
+ \rho\,w_d\,w^d\,\beta_*\,e^H\right)\,
\frac{1 - u}{(1 - \rho)^2}
- \left(\Delta_hu
+ u\,w\,w_d\,w^d\,\beta_*\,e^H\right)
\,\frac{1}{1 - \rho}\right] X_a
\]
\[
= \omega\,X_a - f\,\left[J 
- \frac{1}{3}\, \frac{u - \rho}{1 - \rho}w\,\,w_d\,w^d\,\beta_*\,e^H
\right] X_a = 0.
\]

To extend $Y$ to $U$ we consider the integrability
conditions for the conformal Killing equations (\cite{yano}),
\begin{equation}
\label{2confKil}
D_aD_bY_c = Y_d\,R^d\,_{abc} + \omega_a\,h_{bc} + \omega_b\,h_{ac}
- \omega_c\,h_{ab},
\end{equation}
with
\begin{equation}
\label{Domegadef}
D_a \omega = \omega_a,
\end{equation}
and
\begin{equation}
\label{3confKil}
D_a\,\omega_b = - {\cal L}_Y\,s_{ab}
= - (Y^c\,D_c\,s_{ab} + D_aY^c\,s_{cb} + D_bY^c\,s_{ac}),
\end{equation}
satisfied in our conformal gauge by $Y$ on $V$ and to be satisfied on $U$. 
Fix $p \in V$
and consider the geodesics through
$p$. The equations above imply a linear system of ODE's for $Y_a$,
$D_aY_b$, $\,\omega$, $\,\omega_a$ along the geodesics which determine a unique
analytic extension of $Y$ to $U$. Equations (\ref{KconKil}) and 
(\ref {XKcommute}) extend to $U$ by analyticity.

Because
\[
X^a =  - \frac{l_*}{1 - \rho}\,\epsilon^{abc}\,Y_b\,\rho_c,
\] 
the relation $D_aX_b(i) \neq 0$ implies that $Y^a \neq 0$ at $i$ 
and thus also on the axis near $i$. The restriction of 
\[
0 = [X, Y]_a = X^cD_c Y_a - Y^cD_cX_a,
\]
to the axis implies that $Y$ is tangent to the axis, 
the conformal Killing equation implies that $Y \neq 0$ there.

It holds ${\cal L}_Y\,\tilde{h} = {\cal L}_Y\,(\Omega^{-2}h) = 
2\,\tilde{\omega}\,\tilde{h}$ with
$\tilde{\omega} = \omega - \Omega^{-1}\,Y^aD_a\,\Omega$.
Because $D_a(Y^c\,\rho_c)|_i = - 2\,\mu\,Y_a(i) \neq 0$,
it follows with 
(\ref{Bvomsig})  that $\Omega^{-1}\,Y^aD_a\,\Omega$ diverges at $i$. Thus
$\tilde{\omega}$ can neither be constant nor vanish.
\end{proof}

\vspace{.3cm}

On $V$, which can be chosen to contain $i$ if $u_a(i) \neq 0$, the dualized
version of the Cotton tensor
$B_{bca} = D_{[c}\,R_{a]b}$ acquires by (\ref{Bcott}), (\ref{Ricexpr}),
(\ref{Killingfield}), and  (\ref{conKillingexists}) the concise form
\begin{equation}
\label{Cott-repr}
B_{ab} = \frac{1}{2}\,B_{acd}\,\epsilon_b\,^{cd} 
= \frac{1}{1 - \rho}\,s_{d(a}\,\epsilon_{b)}\,^{cd}\rho_c 
= \frac{\beta_*\,e^{H - 2 L}}{l_*\,(1 - \rho)^5}\,
X_{(a}\,Y_{b)}.
\end{equation}

\vspace{.1cm}

The conformal factor (\ref{Bvomsig}) and the 
transformation laws under conformal rescalings give  
\[
R_{ab}[\tilde{h}] =
\sqrt{\rho}\,\beta\,(w_{a}\,w_b - \frac{1}{3}\,h_{ab}\,w_{c}\,w^c)
- \frac{3}{2\,\sqrt{\rho}^3\,(1 + \sqrt{\rho})^2}
(\rho_a \,\rho_b - \frac{1}{3}\,\rho_c \,\rho^c\,h_{ab}).
\]
At $i$ holds $w_a\,\rho^a = 0$ and
$D_c(w_a\,\rho^a) = 2\,\mu\,u_c(i) \neq 0$ if $u_c(i) \neq 0$.
Thus, if $u_c(i) \neq 0$, there exists a smooth hypersurface $H$
through $i$ on which $w_a\,\rho^a = 0$. 
With $X_a\,\rho^a = 0$, $X_a\,w^a = 0$ it follows that
the vector fields $w^a$, $\rho^a$, $X^a$ define on $H \setminus \{i\}$
an orthogonal set of eigenvectors of $R_{ab}[\tilde{h}]$. The direct
calculation shows that the three eigenvalues are different from each
other. It follows that $\tilde{g}$ is of Petrov type $I$ near $H
\setminus \{i\}$. 
The fields $\rho^a$ and 
$w^a$ are proportional to $u^a$ on the axis because $X = 0$ there.
It follows that 
$R_{ab}[\tilde{h}] \sim \rho_a \,\rho_b - \frac{1}{3}\,\rho_c
\,\rho^c\,h_{ab}$ so that $\tilde{g}$ is of Petrov
type $D$ along the axis (cf. \cite{stephani:et:al}).

\section{Existence and non-existence results}

In spite of the simple conclusion (\ref{nu1ca0}) the case 
were $u_a(i) = 0$ is not easily 
discussed in general. The function 
$\beta = \frac{3}{2}\,\frac{s_{ab}\,w^a w^a}{(w_c w^c)^2}$
in (\ref{2ndbigindcond}) may become, along with the quantities $H$ and $L$,
singular at $i$ if $w_a(i) = 0$. In fact, if $s_{ab}$ does not vanish
at
$i$, it is there due to axi-symmetry of the form
\begin{equation}
\label{sab-at-i}
s_{ab}(i) = \xi\,(n_a n_b + \frac{1}{3}\,h_{ab}), \quad \quad
\xi \neq 0,
\end{equation}
with a unit vector $n^a$ pointing in the
direction of the axis. With this one finds that
$\beta = - sign(\xi)\,\sqrt{3/2\,s_{ed} s^{ed}}\,(w_cw^c)^{-1}$
is unbounded near $i$. If $D_d$ is applied to equations 
(\ref{Bcott}) and (\ref{biguid}), the resulting equations are subtracted from
each other, and the difference is restricted to $i$, one obtains with 
$u_a(i) = 0$ and (\ref{n1fequ}), (\ref{sigmaval}), (\ref{Lapvan}), (\ref{DDvan}) 
at $i$ the relation 
\[
\nu\left\{2\,s_{d[c}s_{a]b} + s_d\,^f s_{f[c}h_{a]b}\right\}
= 
2\,\mu\left\{2\,h_{d[c}s_{a]b} + s_{d[c}h_{a]b}\right\}.
\]
It restricts the parameter in (\ref{sab-at-i}) by 
$\xi = \frac{6\,\mu}{\nu}$ but it does not exclude (\ref{sab-at-i}).
We shall leave the case $D_au(i) = 0$, $s_{ab}(i) \neq 0$ open in this
article. If the second of these conditions is dropped we get with
(\ref{nu1ca0}) a complete, though negative answer.
\begin{lemma} 
\label{du(i)=0}
Suppose $h$, $\rho$ is a solution to the conformal static vacuum
equations (\ref{n1fequ}), (\ref{sigmaval}) such that $s_{ab}(i) = 0$. Then a
positive solution
$u$ to equations  (\ref{Lapvan}), (\ref{DDvan}), (\ref{uati}) near $i$ which
satisfies $u(i) = \nu \neq 1$ and $D_au(i) = 0$ is the constant solution 
$u = \nu$ and $h$ is flat.
\end{lemma}

\begin{proof}  We assume that $u$ is not a constant and show that this leads
to a contradiction. In the following results of
\cite{friedrich:statconv} will be used. For the notions set below
in quotation marks and statements relating to them we refer the reader to
that article.

Let $x^a$ denote $h$-normal coordinates centered at $i$, so that 
$h^*_{ab} \equiv h_{ab}(i) = - \delta_{ab}$.
Equation (\ref{Bcott}) implies with (\ref{sigmaval})
that $D_as_{bc}(i)$ is totally symmetric. More
generally, if
\[
s_{ab} = O(|x|^k) \quad \mbox{for some}  \quad k \ge 1, 
\]
the `exact set of equations argument' implies with (\ref{Bcott}) that the
tensor
\[
\psi_{a_1\,\ldots \,a_k\,a\,b} \equiv D_{a_1}\,\ldots\,D_{a_k}\,s_{ab}(i),
\]
is given in space
spinor notation by a completely symmetric spinor
$\psi_{A_1B_1\,\ldots\,A_{k+2} B_{k+2}}$ because it either vanishes or
defines the non-vanishing `null datum' of lowest order for the solution $h$.
We shall show that it vanishes so that in fact 
$s_{ab} = O(|x|^{k+1})$.

Equations (\ref{Lapvan}),
(\ref{DDvan}), (\ref{contrid}) imply with the assumption above  
\[
u = \nu + O(|x|^{k + 2}), \quad \quad
D_{a_1}\,\ldots\,D_{a_{k+2}}\,u(i) =
- \nu\,(1 - \nu)\,\psi_{a_1\,\ldots \,a_{k+2}},
\]
so that one obtains at $i$ `normal expansions'
\[
s_{ab} = \frac{1}{k!}\,\psi_{a_1\,\ldots
\,a_k\,a\,b}\,x^{a_1}\,\ldots\,x^{a_k} + O(|x|^{k + 1}),
\]
\[
u = \nu - \nu\,(1 - \nu)\,\frac{1}{(k+2)!}\,\psi_{a_1\,\ldots
\,a_{k+2}}\,x^{a_1}\,\ldots\,x^{a_{k+2}} + O(|x|^{k + 3}),
\]
\[
D_au = - \nu\,(1 - \nu)\,\frac{1}{(k+1)!}\,\psi_{a_1\,\ldots
\,a_{k+1}\,a}\,x^{a_1}\,\ldots\,x^{a_{k+1}} + O(|x|^{k + 2}),
\]
\[
\rho = - \mu\,h^*_{ab}\,x^a\,x^b + O(|x|^{k + 4}), \quad
D_a\rho = - 2\,\mu\,h^*_{ab}\,x^b  + O(|x|^{k + 3}),
\]
whence
\[
w_a = \frac{- 2\,\mu\,(1 - \nu)\,h^*_{ab}\,x^b}{(1 +
\mu\,h^*_{ab}\,x^a\,x^b)^2} + O(|x|^{k + 1}), 
\quad
w_a w^a = \frac{(2\,\mu\,(1 - \nu))^2\,h^*_{ab}\,x^a\,x^b}
{(1 + \mu\,h^*_{ab}\,x^a\,x^b)^4} + O(|x|^{k + 2}).
\]
It follows that $w_a \neq 0$ on a punctured
neighbourhood of $i$. Consider near $i$ the real analytic field
\[
f_{ab} = 2\,(w_c w^c)^2 s_{ab} - 3\,w^c w^d\,s_{cd}\,
(w_a\,w_b - \frac{1}{3}\,h_{ab}\,w_e w^e).
\] 
 Lemma \ref{ricciexpr} implies with the assumption that
$u$ is not constant that the field $f_{ab}$
vanishes on the open set where
$w_a \neq 0$ and thus, by analyticity, everywhere. 
In the normal coordinates $x^a$ we extend near $i$ now all real analytic
fields holomorphically into the complex domain and consider a complex null
geodesic $x^a(\tau) = \tau\,l^a$ with $l^a = const. \neq 0$, $l_a l^a = 0$. 

Expanding $f_{ab}(x(\tau))$ at $\tau = 0$  and observing the expansions 
above gives
\[
0 = f_{ab}
= - 3\,\frac{(2\,\mu\,(1 - \nu))^4}{k!}\,
\psi_{a_1\,\ldots\,a_{k+2}}\,l^{a_1}\,\ldots\,l^{a_{k+2}}\,l_a\,l_b\,
\tau^{k+4} +O(|\tau|^{k + 5}).
\]
This shows that
\[
0 = \psi_{a_1\,\ldots\,a_{k+2}}\,l^{a_1}\,\ldots\,l^{a_{k+2}} =
\psi_{A_1B_1\,\ldots\,A_{k+2} B_{k+2}}\,\iota^{A_1}\,\ldots \iota^{B_{k+2}},
\]
where $l^a$ is represented on the right hand side by the spinor
$\iota^A\,\iota^B$. Because $l^a$ and thus the spinor
$\iota^A$ is arbitrary here,  the symmetric spinor
$\psi_{A_1B_1\,\ldots\,A_{k+2} B_{k+2}}$ must vanish, which implies $s_{ab}
= O(|x|^{k+1})$ and $u =
\nu + O(|x|^{k+3})$. In contradiction to our assumption it follows 
inductively that $u - \nu$ and $s_{ab}$ vanish at $i$ at all orders.
\end{proof}

\vspace{.3cm}

The existence result anounced in the introduction will now be proven. In
stating it we ignore trivial  rescalings.

\begin{lemma} 
\label{du(i)neq0}
For given data $\mu, \nu, \beta_* \in \mathbb{R}$ and 
$c^a \in\mathbb{R}^3$ satisfying 
\[
\mu > 0, \quad \nu > 0, 
\quad \beta_* \neq 0, \quad c^a \neq 0, 
\]
there exists a solution $h$ to the conformal static field
equations which admits a non-trivial conformal rescaling with conformal
factor $\vartheta = \sqrt{\nu/u}$. The fields $h$ and $u$ are uniqely
determined by the requirements that 
$m = 2\,\sqrt{\mu}$ is the
ADM mass of the asymptotically flat static metric associated with $h$, the
Ricci tensor of $h$ assumes in $i$-centered $h$-normal coordinates
the value
\[
s_{ab} = \beta_*(c_a\,c_b - \frac{1}{3}\,h_{ab}\,c_d\,c^d) \neq 0
\quad \mbox{at} \quad i,
\]   
and the positive function $u$ satisfies
\[
u = \nu, \quad D_au = c_a
\quad \mbox{at} \quad i.
\]   
\end{lemma}

\begin{proof}
The result will be obtained by solving  simultaneously the conformal field
equations for $h$ and the equations satisfied by $u$.
Let $e_a$, $a = 1, 2, 3$, denote an $h$-orthonomal frame 
and denote the 1-forms dual to it by $\sigma^a$. The metric is then given by
$h = h_{ab}\,\sigma^a\,\sigma^b$
with metric coefficients $h_{ab} = h(e_a, e_b) = -
\delta_{ab}$. The connection coefficients, defined by $D_{e_a}e_b =
\Gamma_a\,^c\,_b\,e_c$, satisfy $\Gamma_{acb} = \Gamma_{a[cb]}$ with
$\Gamma_{acb} = h_{cd}\,\Gamma_a\,^d\,_b$.
The connection form is then given by
$\omega^a\,_b = \Gamma_c\,^a\,_b\,\sigma^c$ so that 
$\omega_{ab} = \omega_{[ab]}$. 
This expansion in terms of the $\sigma^a$ will be
used only later, when we describe the solution
procedure in detail. 

The equations will be written as differential system for the unknown
\[
{\bf U} = (\rho,\,\, \rho_a, \,\, s, \,\, u, \,\, u_a,\,\, t, 
\,\, \beta, \,\, \sigma^a, \,\, \omega^a\,_b),
\]
where the first seven components  denote (vector-valued) 0-forms and 
the last two components are 1-forms. Until we introduce coordinates below
all indices should be understood as frame indices. Consider the
differential forms
\[
\Lambda = d\rho - \rho_a \sigma^a, \quad \quad \quad \quad
\Gamma = d u - u_a \sigma^a,
\]
\[
\Sigma_a = d \rho_a - \rho_c \,\omega^c\,_a - s\,h_{ab}\,\sigma^b
+ \rho\,(1 - \rho)\,s_{ab}\,\sigma^b,
\]
\[
\Pi_a = d u_a - u_c \,\omega^c\,_a - t\,h_{ab}\,\sigma^b
+ u\,(1 - u)\,s_{ab}\,\sigma^b,
\]
\[
S = ds + (1 - \rho)\,\rho^a\,s_{ab}\,\sigma^b, \quad \quad 
T = dt + (1 - u)\,u^a\,s_{ab}\,\sigma^b,
\]
\[
B = d \beta + \beta\,M_a\,\sigma^a,
\]
\[
\Theta^a = d \sigma^a\, + \omega^a\,_b \wedge \sigma^b, \quad \quad 
\Delta^a\,_b = d \omega^a\,_b + \omega^a\,_c \wedge \omega^c\,_b
- \Omega^a\,_b,
\]
denoted collectively by
\[
{\bf \Psi} = (\Lambda,\,\, \Gamma, \,\, \Sigma_a, \,\,\Pi_a, \,\,S, 
\,\,T, \,\, B, \,\,\Theta^a, \,\, \Delta^a\,_b),
\]
or by ${\bf \Psi}^A$ if we need to bring out relations involving
different components. In the differential forms  above and in the forms
derived below we consider functions of the components of ${\bf U}$ which
are given by
\[
w = \frac{1 - u}{1 - \rho}, \quad 
w_a = \frac{1 - u}{(1 - \rho)^2}\,\rho_a - \frac{1}{1 - \rho}\,u_a,
\quad
s_{ab} = \beta\,(w_a\,w_b - \frac{1}{3}\,h_{ab}\,w_cw^c),
\]
\[
M_a = Q\,w_a
+ \frac{1}{3}\,(1 - u)\,(1 - w)\,\beta\,w_a - \frac{1}{1 - \rho}\,\rho_a,
\quad \quad
Q = \frac{5\,(w\,s - t + \rho_c\,w^c)}{w_e w^e\,(1 - \rho)},
\]
\[
\Omega^a\,_b = \frac{1}{2}\,R^a\,_{bcd}\,\sigma^c \wedge \sigma^d
= (h^a\,_c\,s_{bd} - h_{bc}\,s_d\,^a)\,\sigma^c \wedge \sigma^d.
\]

The equations we need to solve read now 
\[
{\bf \Psi} = 0.
\] 
The first two of these equations ensure that 
$\rho_a$ and $u_a$ represent the differentials of $\rho$ and $u$,
the following five equations represent
(\ref{n1fequ}), (\ref{DDvan}), (\ref{B1intco}), (\ref{contrid}), (\ref{betaequ})
and the remaining equations are the first and the second
structural equation.

A lengthy but straight forward calculation shows that 
the differential forms comprised by ${\bf \Psi}$ satisfy
the differential system
\[
d \Lambda = - \Sigma_a \wedge \sigma^a - \rho_a\,\Theta^a,
\]
\[
d \Gamma = - \Pi_a \wedge \sigma^a - u_a\,\Theta^a,
\]

\[
d \,\Sigma_a = - \Sigma_b \wedge \omega^b\,_a - \rho_b\,\Delta^b\,_a
- S \wedge \sigma_a - s\,\Theta_a 
+ (1 - 2\,\rho)\,\Lambda \wedge s_{ab}\,\sigma^b
+ \rho\,(1 - \rho)\,s_{ab}\,\Theta^b
+ \rho\,(1 - \rho)\,{\cal A}_a,
\]
\[
d \Pi_a = - \Pi_b \wedge \omega^b\,_a - u_b\,\Delta^b\,_a 
- T \wedge \sigma_a
- t\,\Theta_a + (1 - 2\,u)\,\Gamma \wedge s_{ab}\,\sigma^b
+ u\,(1 - u)\,s_{ab}\,\Theta^b
+ u\,(1 - u)\,{\cal A}_a,
\]

\[
d T = - \Lambda \wedge \rho^a\,s_{ab}\,\sigma^b 
+ (1 - \rho)\,\Sigma_a \wedge s^a\,_b\,\sigma^b
+ (1 - \rho)\,\rho^a\,s_{ab}\,\Theta^b
+ (1 - \rho)\,\rho^a\,{\cal A}_a,
\]
\[
d S = - \Gamma \wedge u^a\,s_{ab}\,\sigma^b 
+ (1 - u)\,\Pi_a \wedge s^a\,_b\,\sigma^b
+ (1 - u)\,u^a\,s_{ab}\,\Theta^b 
+ (1 - u)\,u^a\,{\cal A}_a,
\]

\[
d B = B \wedge (M_a\,\sigma^a + \beta\,F\,w_a\,\sigma^a) 
\]
\[
+\, \beta\left\{\frac{Q}{1 - \rho}\,\Lambda
+ \left(\frac{5\,\rho^c}{(1 - \rho)^3\,w_bw^b}
- \frac{2\,Q\,w^c}{(1 - \rho)^2\,w_b\,w^b}\right)
\left(- \rho_c\,\Gamma
+ w\,\rho_c\,\Lambda + 
(1 - \rho)\,(w_c\,\Lambda
+ w\,\Sigma_c - \Pi_c)\right)\right.
\]
\[
\left. 
+ \frac{5}{(1 - \rho)^2\,w_b w^b}\left(
s\,w\,\Lambda - s\,\Gamma 
+ (1 - \rho)\,(w\,S - T + w^c\,\Sigma_c)
\right)\right\} \wedge w_a \sigma^a
\]
\[
+ \,\beta\left\{ \frac{Q + \beta\,F}{(1 - \rho)^2}
\left(w\,\rho_a\,\Lambda - \rho_a\,\Gamma
+ (1 - \rho)\,(w_a\,\Lambda + w\,\Sigma_a - \Pi_a)
\right) - \frac{1}{3}\,\beta\,((1 - 2\,w)\Gamma + w^2\,\Lambda)\,w_a
\right.
\]
\[
\left. - \frac{\rho_a}{(1 - \rho)^2}\,\Lambda - \frac{1}{1 - \rho}\,\Sigma_a
\right\} \wedge \sigma^a + \beta\,M_a\,\Theta^a,
\]

\[
d \Theta^a = \Delta^a\,_b \wedge \sigma^b - \omega^a\,_b \wedge \Theta^b, 
\]
\[
d \Delta^a\,_b = \Delta^a\,_c \wedge \omega^c\,_b 
- \omega^a\,_c \wedge \Delta^c\,_b
- (h^a\,_c\,s_{bd} - h_{bc}\,s_d\,^a)\,(\Theta^c \wedge \sigma^d
- \sigma^c \wedge \Theta^d)
+ \sigma^a \wedge {\cal A}_b - {\cal A}^a \wedge \sigma_b,
\]

where
\[
{\cal A}_a = 
w_a\,B \wedge w_b\,\sigma^b - \frac{1}{3}\,w_b w^b\,B \wedge \sigma_a
+ \beta\,({\cal B}_a \wedge w_b \,\sigma^b + w_a\,{\cal B}_b \wedge \sigma^b
- \frac{2}{3}w^b\,{\cal B}_b \wedge \sigma_a),
\]
with
\[
{\cal B}_a = - \frac{\rho_a}{(1 - \rho)^2}\,\Gamma
+ \left(\frac{w_a}{1 - \rho} + \frac{w\,\rho_a}{(1 - \rho)^2}\right)\Lambda
+ \frac{w}{1 - \rho}\,\Sigma_a
- \frac{1}{1 - \rho}\,\Pi_a.
\]

In deriving the last equation there arises the expression
\[
{\cal N}^{ab} = - w_d\,w^d\,\sigma^a \wedge \rho_c\,\sigma^c \wedge \sigma^b
+ \rho_c\,w^c\,\sigma^a \wedge w_d\,\sigma^d \wedge \sigma^b
+ \sigma^a \wedge \rho_c\,\sigma^c \wedge w_d\,\sigma^d \,\,w^b
- w^a\,\,\rho_c\,\sigma^c \wedge w_d\,\sigma^d \wedge \sigma^b.
\]
By writing 
$\sigma^a \wedge \sigma^b \wedge \sigma^c = \epsilon^{abc}\,\mu_h$ with the
volume form defined by $h$ on the right hand side and calculating
${\cal N}^{ab}\,\epsilon_{abc}$, the field ${\cal N}^{ab}$  
can be shown to vanish identically.

In short notation (the summation convention applying to all indices)
the relations above take the form
\begin{equation}
\label{dPsirel}
d\,{\bf \Psi}^A = f^A_{\,B\,c}\,\sigma^c \wedge {\bf \Psi}^B
+ f^{A\,a}_{\,B\,b}\,\omega^b\,_a \wedge {\bf \Psi}^B
+ f^A_{\,B}\,{\bf \Psi}^B, 
\end{equation}
with functions
$f^A_{\,B\,c}$, $f^{A\,a}_{\,B\,b}$, $f^A_{\,B}$ of the
$0$-forms given by ${\bf U}$.

\vspace{.3cm}

Because we assume that $D_au(i) = c_a \neq 0$, all fields given by 
$\bf U$, in particular $\beta$ (with the meaning given to it earlier), 
can be assumed to be smooth near $i$. The solutions to the  equations
${\bf \Psi} = 0$ are obtained near $i$ as follows. 
Assume that $x^a$ denote $i$-centered, $h$-normal coordinates and $e_a$ an
$h$-orthonormal normal frame centered at $i$ so that 
$e^a\,_b = \,<dx^a,\,e_b>\, = \delta^a\,_b$ at $i$. 
Then $\hat{Y} = x^a\,\partial_a$ is smooth near $i$ and so that 
$Y = \frac{1}{|x|}\,\hat{Y}$ is for $x^a \neq 0$ the geodesic unit vector
field tangent to the geodesics through $i$, which has direction dependent
limits as $|x| \rightarrow 0$.
Writing 
$\sigma^a = \sigma^a\,_b\,dx^b$, we find the inner products
\begin{equation}
\label{ihngauge}
i_{\hat{Y}}\,\sigma^a = \sigma^a\,_b(x)\,x^b = x^a, \quad \quad 
i_{\hat{Y}}\,\omega^a\,_b = x^c\,\Gamma_c\,^a\,_b(x) = 0
\quad \mbox{near} \quad x^a = 0.
\end{equation}

Consider now a radial geodesic $\gamma: \tau \rightarrow \tau\,x^a_*$ with 
some $x^a_*$ satisfying $\delta_{ab}\,x^a_*\,x^b_* = 1$. 
The equation
\[
i_{Y}\,\Theta^a = 0,
\]
implies then for the unknown
$\hat{\sigma}^a\,_b = \sigma^a\,_b - \delta^a\,_b$,  
which is required to be smooth, the initial condition
\begin{equation}
\label{sigmahatati}
\hat{\sigma}^a\,_b \rightarrow 0 \quad \mbox{as} \quad 
\tau \rightarrow 0,
\end{equation} 
and along $\gamma$ the ODE
\begin{equation}
\label{firststrequODE}
\frac{d}{d\tau}\,\hat{\sigma}^a\,_b + \frac{1}{\tau}\,\hat{\sigma}^a\,_b
= \Gamma_c\,^a\,_d\,\hat{\sigma}^c\,_b\,x^d_*
+ \Gamma_b\,^a\,_d\,x^d_*,
\end{equation}
where we expressed the connection form in terms of the connection
coefficients, which are considered now as unknowns.
The equation
\[
i_{Y}\,\Delta^a\,_b = 0, 
\]
similarly implies the initial condition and the ODE 
\begin{equation}
\label{Gammaati}
\Gamma_c\,^a\,_b  \rightarrow 0 \quad \mbox{as} \quad 
\tau \rightarrow 0,
\end{equation}
\begin{equation}
\label{secstrequODE}
\frac{d}{d\tau}\,\Gamma_c\,^a\,_b 
+ \frac{1}{\tau}\,(\delta^d\,_c - \hat{\sigma}^d\,_f\,e^f\,_c)\,
\Gamma_d\,^a\,_b 
= - \Gamma_d\,^a\,_b\,\Gamma_c\,^d\,_e\,x^e_*
+ R^a\,_{bdc}\,x^d_*,
\end{equation}
where the frame coefficients $e^a\,_b$ (by
Cramer's rule rational functions of the $\sigma^a\,_b$)
satisfy
\[
\sigma^a\,_c\,e^c\,_b = \delta^a\,_b \quad \mbox{so that} \quad
e^a\,_b = \delta^a\,_b - \hat{\sigma}^a\,_c\,e^c\,_b =
\delta^a\,_b + O(|\tau|).
\]
The equations
\begin{equation}
\label{regODE}
i_{Y}\,\Lambda = 0,\,\,\, i_{Y}\,\Gamma = 0, \,\,\,
i_{Y}\,\Sigma_a = 0, \,\,\,i_{Y}\,\Pi_a = 0,\,\,\,
i_{Y}\,S = 0, \,\,\,i_{Y}\,T = 0, \,\,\,i_{Y}\,B = 0,
\end{equation}
imply regular ODE's along $\gamma$. The unknowns in these equations must
satisfy the initial conditions
\begin{equation}
\label{Uati}
\rho = 0,\,\, \rho_a = 0, \,\, s = - 2\,\mu, \,\, u = \nu > 0, \,\, 
u_a = c_a \neq 0,\,\,t = \frac{c_a c^a}{2\,\nu}, 
\,\, \beta = \beta_* \neq 0
\quad \mbox{at} \quad \tau = 0.
\end{equation} 

{\it There exists along $\gamma$ a unique solution ${\bf U}$ to the
system consisting of (\ref{firststrequODE}), (\ref{secstrequODE}), and the ODEs
implied by (\ref{regODE})  which satisfies the conditions
(\ref{ihngauge}), (\ref{sigmahatati}),
(\ref{Gammaati}), (\ref{Uati}). This solution is real analytic in the initial
data (\ref{Uati}) and $\tau\,x^a_*$. The fields given by ${\bf U}$ are in fact
analytic in the coordinates $x^a$ and satify the equation ${\bf \Psi} = 0$
near $x^a = 0$}.

This result follows from general properties of systems of ODE's.
The only subtlety arises here from the singularity of equations
(\ref{firststrequODE}), (\ref{secstrequODE}) at $\tau = 0$. It is such that 
the left hand sides of these equations are of the form
$\dot{x} + \tau^{-1}\,A\,x$ with a matrix $A$
which approaches a diagonal matrix with positive entries as
$\tau \rightarrow 0$. The existence of an unique analytic 
solution to the complete system then follows immediately with the
methods used in \cite{friedrich:statconv} where ODE's of the same type
have been discussed.

Once the solution ${\bf U}$ has been obtained, equations
(\ref{firststrequODE}), (\ref{secstrequODE}) imply a system of ODE's for
$x^a_*\,\hat{\sigma}^a\,_b$ and $x^a_*\,\Gamma_a\,^b\,_c$ along $\gamma$
which allows one to concluce that these quantities vanish everywhere so that
the relations (\ref{ihngauge}) are indeed satisfied. 
Because $i_{Y}\,{\bf \Psi} \rightarrow 0$ as $x^a
\rightarrow 0$ along any geodesic passing through $x^a = 0$ it
follows that ${\bf \Psi} = 0$ at $x^a = 0$. 
From (\ref{dPsirel}), (\ref{ihngauge}), and $i_{Y}\,{\bf \Psi} = 0$ one obtains 
\[
{\cal L}_{Y} {\bf \Psi}^A = (d \circ i_{Y} + i_{Y} \circ d) {\bf \Psi}^A
= i_{Y} \,d\,{\bf \Psi}^A
= f^A_{\,B\,c}\,\frac{x^c}{|x|} \,{\bf \Psi}^B.
\]
This equation implies for $|x| \neq 0$ along each geodesic $\gamma$ passing
through $x^a = 0$ a linear homogeneous system of ODE's for the coefficient
functions defining the forms given by ${\bf \Psi}$. It behaves
regularly on $\gamma$ as $|x| \rightarrow 0$. 
Since ${\bf \Psi}$ vanishes at $x^a = 0$, the asserted result follows.
\end{proof}

\vspace{.3cm}

In counting the free parameters in lemma \ref{du(i)neq0}, the vector $c^a$
should be taken into account only in terms of the number $c = \sqrt{c_ac^a}
> 0$ because of the freedom to rotate the normal coordinates around
their origin. Ignoring the parameter $\mu$, which can be changed by trivial
rescalings, it follows the that solutions depend on three parameters.

\section{Concluding remarks}

The statements of theorem \ref{exist3parameter-set} represent an
extract from the results proven in the lemmas and the subsequent remarks. The
latter provide much more information on the various structures. This will
become important in a sequel to this article in which the properties of the
solutions whose existence has been obtained in theorem
\ref{exist3parameter-set} will be discussed in some detail.

}

\end{document}